\documentclass[12pt]{amsart}

\usepackage{setspace}
\usepackage{comment}
\usepackage{graphicx}
\usepackage{caption}
\usepackage{subcaption}
\usepackage{natbib}
\usepackage{color}
\usepackage{array}

\usepackage[foot]{amsaddr}

\usepackage{anysize}
\marginsize{2cm}{2cm}{2cm}{3cm}

\newtheorem{thm}{Theorem}

\newtheorem{lem}[thm]{Lemma}
\newtheorem{prop}[thm]{Proposition}
\newtheorem{defn}[thm]{Definition}

\newtheorem{rem}[thm]{Remark}

\numberwithin{thm}{section}
\numberwithin{equation}{section}
\numberwithin{figure}{section}
\numberwithin{table}{section}

\newcommand{\R}{\mathbb{R}}
\newcommand{\E}{\mathbb{E}}
\newcommand{\TM}{\mathrm{TM}}
\newcommand{\DT}{\mathrm{DT}}

\newcommand{\VAR}{{\rm VaR}}
\newcommand{\ES}{{\rm ES}}

\newcommand{\Var}{{\rm Var}}

\newcommand{\rc}{{\rm RC}}
\newcommand{\mrc}{\mathrm{MRC}}
\newcommand{\frc}{\mathrm{FRC}}
\newcommand{\CED}{\mathrm{CED}}

\newcommand{\corr}{\rm Corr}

\newcommand{\Prob}{\mathbb{P}}
\newcommand{\F}{\mathcal{F}}
\newcommand{\cadlag}{c\`{a}dl\`{a}g }
\newcommand{\Rspace}{\mathcal{R}}

\newcommand{\norm}[2]{\| #1 \|_{#2}}

\begin{document}

\title[Drawdown: From Practice to Theory and Back Again]{Drawdown: From Practice to Theory and Back Again \\ \vspace{5pt} (\lowercase{\normalfont forthcoming in} {\emph{\scshape \itshape \small Mathematics and Financial Economics}})} 
\author{Lisa R. Goldberg$^1$}
\address{$^1$Department 
of Statistics and Economics and Center for Risk Management Research, University of California, 
Berkeley, CA 94720-3880, USA}
\email{$^1$lrg@berkeley.edu}
\author{Ola Mahmoud$^2$}
\address{$^2$Faculty of Mathematics and Statistics, University of St. Gallen, Bodanstrasse 6, CH-9000, Switzerland and Center for Risk Management Research, University of California, Berkeley, Evans Hall, CA 94720-3880, USA}
\email{$^2$olamahmoud@berkeley.edu}
\thanks{We are grateful to Robert Anderson for insightful comments on the material discussed in this article; to Alexei Chekhlov, Stan Uryasev, and Michael Zabarankin for their feedback on a previous draft of this work; to Vladislav Dubikovsky, Michael Hayes, and M\'{a}rk Horv\'{a}th for their contributions to an earlier version of this article; to Carlo Acerbi for providing detailed comments on a previous draft; and to the referees and editors of \emph{Mathematics and Financial Economics} for their valuable feedback.}
\date{ \today}
\maketitle

\begin{abstract}
Maximum drawdown, the largest cumulative loss from peak to trough, is one of the most
widely used indicators of risk in the fund management industry, but one of the least developed in
the context of measures of risk. We  formalize drawdown risk  as Conditional
Expected Drawdown (CED), which is the tail mean of maximum drawdown distributions.
We show that CED is a degree one positive homogenous risk measure, so that it can be linearly attributed to factors;   and convex, so that it can be used in quantitative optimization. We empirically explore the differences in risk attributions based on CED,
Expected Shortfall (ES) and volatility. An important feature of CED is its sensitivity to serial correlation.
In  an empirical study that fits AR(1) models to  US Equity and US Bonds, we find substantially higher correlation between the autoregressive parameter and CED than with ES or with volatility.

\bigskip
\noindent \emph{Key terms:} drawdown; Conditional Expected Drawdown; deviation measure; risk attribution;  serial correlation 

\bigskip
\noindent \emph{Disclosure of potential conflicts of interest:} the authors declare that they have no conflict of interest.

\end{abstract}

\newpage 


\section{Introduction}

A levered investor is liable to get caught in a liquidity trap: unable to secure funding after an abrupt market decline, he may be forced to sell valuable positions under unfavorable market conditions. This experience was commonplace during the 2007-2009 financial crisis and it has refocused the attention of both levered and unlevered investors on an important liquidity trap trigger, a drawdown, which is the maximum decline in portfolio value over a fixed horizon (see Figure \ref{drawdown}). 

\begin{figure}
\centering
  \includegraphics[width=0.7\linewidth]{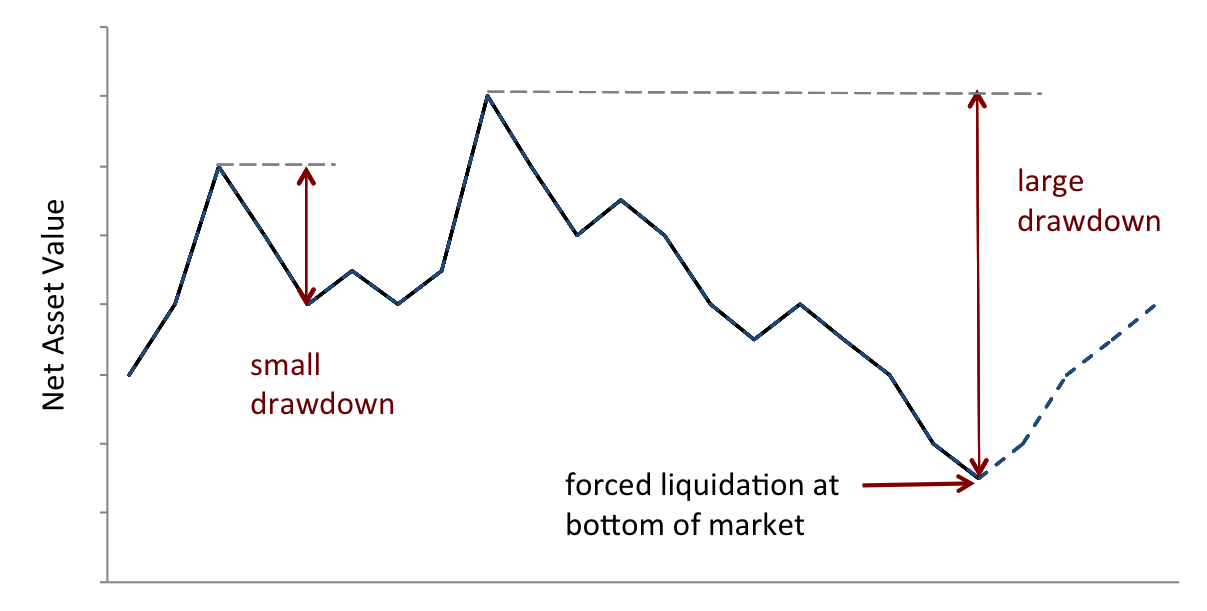}
  \caption{\onehalfspacing Simulation of a portfolio's net asset value over a finite path. A large drawdown may force liquidiation at the bottom of the market, and the proceeding market recovery is never experienced.\\}
  \label{drawdown}
\end{figure}

In the event of a large drawdown, common risk diagnostics, such as volatility, Value-at-Risk, and Expected Shortfall, at the end of the intended investment horizon are less significant. Indeed, within the universe of hedge funds and commodity trading advisors (CTAs), one of the most widely quoted measures of risk is maximum drawdown. The notion of drawdown has been extensively studied in the literature of applied probability theory, which we review in Section~\ref{section:litreview}. However, a generally accepted mathematical methodology for forming expectations about future potential maximum drawdowns does not seem to exist in the investment management industry. Drawdown in the context of risk and deviation measures has failed to attract the same kind of applied research devoted to other more conventional risk measures.

Our purpose is to formulate a (i) mathematically sound and (ii) practically useful measure of drawdown risk. Our formalization of drawdown risk is achieved by modeling continuous-time cumulative returns within a time horizon $T \in (0,\infty)$ as a stochastic process $X$ representing return paths, to which a certain real-valued functional, the \emph{Conditional Expected Drawdown}, is applied. Mathematically, the process $X$ is transformed to the random variable $\mu(X)$, representing the maximum drawdown within a finite path. At confidence level $\alpha \in [0,1]$, the \emph{Conditional Expected Drawdown} CED$_\alpha$ is then defined to be the expected maximum drawdown given that some maximum drawdown threshold $\rm{DT}_\alpha$, the $\alpha$-quantile of the maximum drawdown distribution, is breached:
$$ \rm{CED}_{\alpha}(X)=\mathbb{E}\left(\mu(X) \mid \mu(X)>\rm{DT}_{\alpha}\right).$$

In the context of quantitative risk measures, CED  is a deviation measure in the sense of \citet{Rockafellar2002, Rockafellar2006}. 
In particular, this implies that CED is convex with respect to portfolio weights, which 
means that it promotes diversification and can be used in an optimizer.  It is also homogenous of degree one, so that it supports linear risk attribution under Euler's homogenous function theorem.
 
By focusing on the \emph{maximum} of all drawdowns within a path of fixed length $T$, we address a highly relevant risk management concern affecting fund managers on a daily basis, who ask themselves: what is the expected maximum possible cumulative drop in net asset value within the investment horizon $T$? If this loss exceeds a certain threshold, the investor may be forced to liquidate. For a given investment horizon $T$, Conditional Expected Drawdown indicates this expected cumulative loss in excess of a threshold, and it can be measured for various confidence levels.

Because Conditional Expected Drawdown is defined as the tail mean of a distribution of maximum drawdowns, it is a downside risk metric perfectly analogous to Expected Shortfall, which is the tail mean of a return distribution. Hence, much of the theory and practice of Expected Shortfall carries over to Conditional Expected Drawdown.

We will show, however, that drawdown is inherently path dependent and accounts for serial correlation, whereas Expected Shortfall does not account for consecutive losses.

\subsection{Literature Review}
\label{section:litreview}

The notion of drawdown has been extensively studied in the literature of applied probability theory and in research addressing active portfolio management, which we review next. However, a generally accepted mathematical methodology for forming expectations about future potential maximum drawdowns does not seem to exist in either the investment management industry or the academic literature. Drawdown in the context of risk and deviation measures has hence failed to attract the same kind of applied research devoted to other more conventional risk measures.
Our work hence complements the existing literature as it develops a mathematically sound and practically useful measure of drawdown risk. 

The analytical assessment of drawdown magnitudes has been broadly studied in the literature of applied probability theory. To our knowledge, the earliest mathematical analysis of the maximum drawdown of a Brownian motion appeared in \citet{Taylor1975}, and it was shortly afterwards generalized to time-homogenous diffusion processes by \citet{Lehoczky1977}. \citet{Douady2000} and \citet{Magdon2004} derive an infinite series expansion for a standard Brownian motion and a Brownian motion with a drift, respectively. The discussion of drawdown magnitude was extended to studying the frequency rate of drawdown for a Brownian motion in \citet{Landriault2015}. Drawdowns of spectrally negative L\'{e}vy processes were analyzed in \citet{Mijatovic2012}. The notion of drawup, which measures the maximum cumulative gain relative to a running minimum, has also been investigated probabilistically, particularly in terms of its relationship to drawdown; see for example \citet{Hadjiliadis2006}, \citet{Pospisil2009}, and \citet{Zhang2010}.

Reduction of drawdown in active portfolio management has received considerable attention in mathematical finance research. \citet{Grossman1993} considered an asset allocation problem subject to drawdown constraints; \citet{Cvitanic1995} extended the same optimization problem to the multi-variate framework; \citet{Chekhlov2003, Chekhlov2005} developed a linear programming algorithm for a sample optimization of portfolio expected return subject to constraints on drawdown, which, in \citet{Krokhmal2003}, was numerically compared to shortfall optimzation with applications to hedge funds in mind; \citet{Carr2011} introduced a new European style drawdown insurance contract and derivative-based drawdown hedging strategies; and most recently \citet{Cherney2013}, \citet{Sekine2013}, \citet{Zhang2013} and \citet{Zhang2015} studied drawdown optimization and drawdown insurance under various stochastic modeling assumptions.
\citet{Uryasev2014a} reformulated the necessary optimality conditions for a portfolio optimization problem with drawdown in the form of the Capital Asset Pricing Model (CAPM), which is used to derive a notion of drawdown beta. More measures of sensitivity to drawdown risk were introduced in terms of a class of drawdown Greeks in \citet{Pospisil2010}.
 
In the context of quantitative risk measurement, \citet{Chekhlov2003, Chekhlov2005} develop a quantitative measure of drawdown risk called Conditional Drawdown at Risk (CDaR). Like CED, CDaR is a deviation measure (\citet{Rockafellar2002, Rockafellar2006}). Unlike CED, however, CDaR focuses on all drawdowns rather than maximum drawdowns.


\section{Measuring Drawdown Risk}
\label{section:measuring_ced}

We use the general setup of \citet{Cheridito2004} for the mathematical formalism of continuous-time path dependent risk. Continuous-time cumulative returns, or equivalently net asset value processes, are represented by essentially bounded \cadlag processes (in the given probability measure) that are adapted to the filtration of a filtered probability space. More formally, for a time horizon $T \in (0,\infty)$, let $(\Omega, \F, \{\F_t\}_{t\in[0,T]},\Prob)$ be a filtered probability space satisfying the usual assumptions, that is the probability space $(\Omega, \F,\Prob)$ is complete, $(\F_t)$ is right-continuous, and $\F_0$ contains all null-sets of $\F$. For $p \in [1,\infty]$,  $(\F_t)$-adapted \cadlag processes lie in the Banach space
$$ \Rspace^p = \left\{  X \colon [0,T] \times \Omega \to \R \mid X  \hspace{5pt} (\F_t)\textrm{-adapted \cadlag process }, \norm{X}{\Rspace^p} \right\}, $$
which comes equipped with the norm
$$ \norm{X}{\Rspace^p} := \norm{X^*}{p} $$
where $X^* = \sup_{t\in[0,T]} |X_t|$.

All equalities and inequalities between processes are understood throughout in the almost sure sense with respect to the probability measure $\Prob$. For example, for processes $X$ and $Y$, $X \leq Y$ means that for $\Prob$-almost all $\omega \in \Omega$, $X_t(\omega) \leq Y_t(\omega)$ for all $t$.

\begin{defn}[Continuous-time path-dependent risk measure]
A continuous-time \emph{path-dependent risk measure} is a real-valued function $\rho: \Rspace^\infty \to \R$.\\
\end{defn}
 
In practice, where one works in a discrete universe, this continuous-time setup is discretized by choosing the frequency of observations over the return horizon $T$. This adds a crucial parameter to the analysis, as higher frequency observations tend to yield larger drawdowns. Consider the May 2011 flash crash. When working at a daily frequency, one never sees the flash crash drawdown, no matter how long the investment horizon.\footnote{See \cite{Madhavan2012} for an analysis of the flash crash.}

\subsection{Maximum Drawdown}

\begin{defn}[Drawdown process]
For a horizon $T\in(0,\infty)$, the \emph{drawdown process} $D^{(X)}:=\{D^{(X)}_t\}_{t\in[0,T]}$ corresponding to a stochastic process $X\in\Rspace^\infty$ is defined by  
$$ D^{(X)}_t = M^{(X)}_t - X_t \ ,$$
where 
$$M^{(X)}_t = \sup_{u\in[0,t]} X_u$$
is the running maximum of $X$ up to time $t$.
\end{defn}

In practice, the use of the maximum drawdown as an indicator of risk is particularly popular in the universe of hedge funds and commodity trading advisors, where maximum drawdown adjusted performance measures, such as the Calmar ratio, the Sterling ratio and the Burke ratio, are frequently used. 

\begin{defn}[Maximum drawdown]
Within a fixed time horizon $T \in (0,\infty)$, the \emph{maximum drawdown} of the stochastic process $X\in\Rspace^\infty$ is the maximum drop from peak to trough of $X$ in $[0,T]$, and hence the largest amongst all drawdowns $D^{(X)}_t$:
$$ \mu (X) = \sup_{t\in[0,T]}\{D^{(X)}_t\} .$$
Equivalently, maximum drawdown can be defined as the random variable obtained through the following transformation of the underlying stochastic process $X$:
$$\mu(X) = \sup_{t\in[0,T]} \sup_{s\in[t,T]} \left\{ X_s - X_t\right\} \ .$$
\end{defn}

Even though, in a given horizon, only a single maximum drawdown is realized along any given path, it is beneficial to consider the distribution from which the maximum drawdown is taken. By looking at the maximum drawdown distribution, one can form reasonable expectations about the size and frequency of maximum drawdowns for a given portfolio over a given investment horizon. 

Figure \ref{maxdd_dist} shows (A) the empirical maximum drawdown distribution (for paths of length 125 business days) of the daily S\&P 500 time series over the period 1950 to 2013, and (B) the simulated distribution for an idealized Gaussian random variable. Both distributions are asymmetric, which implies that very large drawdowns occur less frequently than smaller ones. 
Using Monte Carlo simlations, \cite{Burghardt2003} show that maximum drawdown distributions are highly sensitive to the length of the track record\footnote{The track record is understood as the length of the history of an investment fund since its inception.} (increases in the length of the track record shift the entire distribution to the right), mean return (for larger mean returns, the distribution is less skewed to the right, since large means tend to produce smaller maximum drawdowns, volatility of returns (higher volatility increases the likelihood of large drawdowns), and data frequency (a drawdown based on lower frequency data would ignore the flash crash). 

The tail of the maximum drawdown distribution, from which the likelihood of a drawdown of a given magnitude can be distilled, is of particular interest in practice. Our drawdown risk metric, defined next, is a tail mean of the maximum drawdown distribution.

\begin{figure}
\centering
\begin{subfigure}{.5\textwidth}
  \centering
  \includegraphics[width=\linewidth]{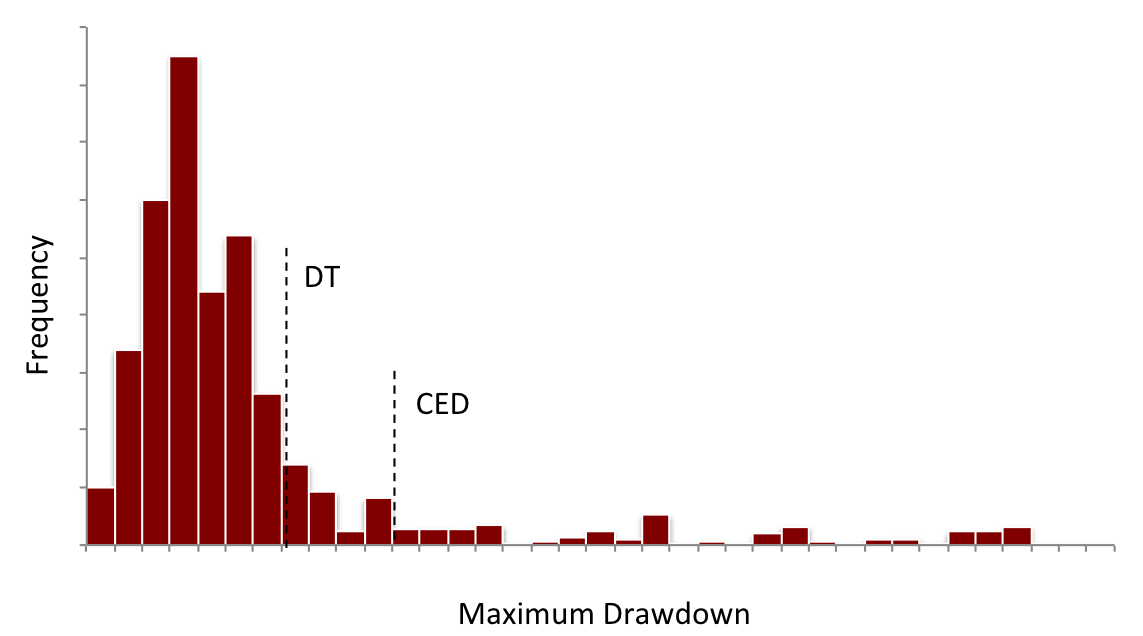}
  \caption{}
  \label{a}
\end{subfigure}%
\begin{subfigure}{.5\textwidth}
  \centering
  \includegraphics[width=\linewidth]{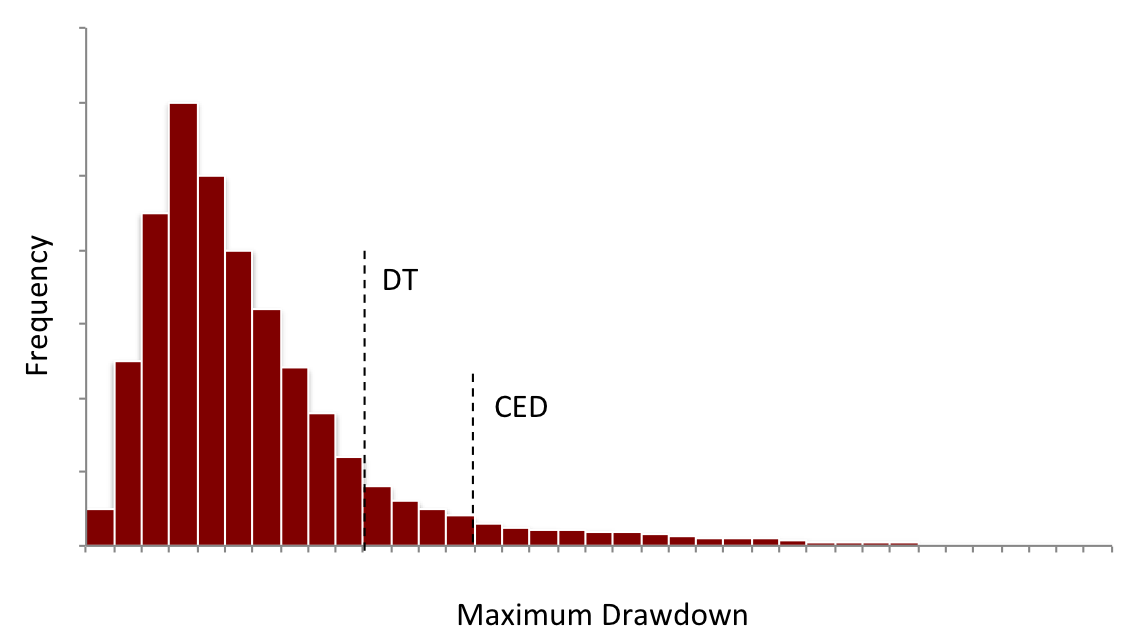}
  \caption{}
  \label{b}
\end{subfigure}
\caption{\onehalfspacing (A) Empirical distribution of the realized 6-month maximum drawdowns for the daily S\&P 500 over the period 1 January 1950 to 31 December 2013, together with the 90\% quantile (the \emph{drawdown threshold} DT) and tail-mean (CED) of the distribution. (B) Distribution of 6-month maximum drawdowns for an idealized standard normally distributed random variable, together with the 90\% quantile and tail-mean of the distribution.}
\label{maxdd_dist}
\end{figure}

\subsection{Conditional Expected Drawdown}

Our proposed drawdown risk metric, the \emph{Conditional Expected Drawdown} (Definition~\ref{defn:ced}), measures the average of worst case maximum drawdowns exceeding a quantile of the maximum drawdown distribution. Hence, it is analogous to the return-based Expected Shortfall ($\ES$). Both ES and CED are given by the tail mean of an underlying distribution, namely that of the losses and maximum drawdowns, respectively.

 Analogous to the return-based Value-at-Risk (VaR), we define, for confidence level $\alpha \in [0,1]$, the \emph{maximum drawdown threshold} $\rm{DT}_\alpha$ to be a quantile of the maximum drawdown distribution:
$$ \DT_{\alpha} \left( \mu(X) \right)= \inf \left\{ m  \mid  \mathbb{P} \left( \mu(X) > m \right) \leq 1-\alpha \right\} $$
It is thus the smallest maximum drawdown $m$ for which the probability that the maximum drawdown $\mu(X)$ exceeds $m$ is at most $(1-\alpha)$. 
For example, the 95\% maximum drawdown is both a worst case for drawdown in an ordinary period and a best case among extreme scenarios. It separates the 5\% worst maximum drawdowns from the rest.  

\begin{defn}[Conditional Expected Drawdown]
At confidence level $\alpha \in [0,1]$, the \emph{Conditional Expected Drawdown} CED$_\alpha:\Rspace^\infty\to\R$ is the function mapping $\mu(X)$ to the expected maximum drawdown given that the maximum drawdown threshold at $\alpha$ is breached. More formally, 
$$ \CED_\alpha \left(X \right) = \frac{1}{1-\alpha} \int_\alpha^1 \rm{DT}_\mathit{u}\left( \mu(X) \right) d \mathit{u} .$$
If the distribution of $\mu(X)$ is continuous, then $\CED_\alpha$ is equivalent to the tail conditional expectation:
$$ \rm{CED}_{\alpha}\left(X \right) = \mathbb{E}\left(\mu(X) \mid \mu(X) > \rm{DT}_{\alpha}\left( \mu(X) \right)\right).$$
\label{defn:ced}
\end{defn}

In other words, CED is the \emph{tail mean} (\cite{AcerbiTasche2002a}) over the maximum drawdown distribution, where for confidence level $\alpha \in (0,1)$, and assuming $\E[\mu(X)] < \infty$, the \emph{$\alpha$-tail mean} of $\mu(X)$ is given by: 
$$\TM_\alpha(\mu(X)) = \frac{1}{1-\alpha} \int_\alpha^1 \DT_u(\mu(X)) d\mathit{u}\ . $$


\section{Properties of Conditional Expected Drawdown}
\label{section:axioms}

We derive theoretical properties of Conditional Expected Drawdown, most notably convexity and positive homogeneity, and prove that it is a \emph{generalized deviation measure}, as developed by \cite{Rockafellar2002, Rockafellar2006}. Broadly speaking, deviation measures obey axioms taken from the properties of measures such as standard deviation and semideviation. We generalize these axioms to our path-dependent universe.

\begin{defn}[Generalized Path-Dependent Deviation Measure]
A generalized path-dependent \emph{deviation measure} is a path-dependent risk measure $ \delta : \Rspace^\infty \to \R$ satisfying the following axioms:
\begin{itemize} 
\item[(D0)] Normalization: for all constant deterministic $C \in \Rspace^\infty$, $\delta(C) = 0$.
\item[(D1)] Positivity: for all $X \in \Rspace^\infty$, $\delta(X) \geq 0$.
\item[(D2)] Shift invariance: for all $X \in \Rspace^\infty$ and all constant deterministic $C \in \Rspace^\infty$, $\delta(X+C) = \delta(X)$. 
\item[(D3)] Convexity: for all $X, Y \in \Rspace^\infty$ and $\lambda \in [0,1]$, $\delta(\lambda+ (1-\lambda)Y) \leq \lambda \delta(X) + (1-\lambda)\delta(Y)$.  
\item[(D4)] Positive degree-one homogeneity: for all $X \in \Rspace^\infty$ and $\lambda > 0$, $\delta(\lambda X) = \lambda \delta(X)$. 
\end{itemize} 
\end{defn} 

Any portfolio of zero value and, more generally, of constant deterministic value is not exposed to drawdown risk, and so for all constant deterministic $C \in \Rspace^\infty$, we have $\CED_\alpha(C) = 0$, and hence axiom (D0) is satisfied. Moreover, CED satisfies (D1) because maximum drawdown is by definition non-negative. The following Lemma proves the shift invariance property (D2), 
which essentially states that by (deterministically) shifting the path of the portfolio value up or down, the drawdown within that path remains unchanged. 

\begin{lem}\label{lem:trans_inv}
 For all $X \in \Rspace^\infty$ and all constant almost surely $C \in \Rspace^\infty$, $\CED_\alpha(X+C) = \CED_\alpha(X)$ (for all $\alpha \in (0,1)$). 
\end{lem}

\begin{proof}
The drawdown process $D^X$ corresponding to $X$ is shift invariant, since for $t\in[0,T]$, 
$$M_t^{(X+C)} = \sup_{u\in[0,t]}(X+C)_u = \sup_{u\in[0,t]}(X)_u+C = M_t^{(X)}+C\ .$$ 
It follows that
$D^{(X+C)} = M^{(X+C)} - X- C = M^{(X)}+C-X-C = M^{(X)}-X = D^{(X)}$. 
Therefore, 
$$\mu(X+C) = \sup_{t\in[0,T]}\left\{  D_t^{(X+C)}\right\} = \sup_{t\in[0,T]}\left\{ D_t^{(X)}\right\} = \mu(X)\ .$$
Hence, $\CED_\alpha(X+C) = \CED_\alpha(X)$.\\
\end{proof}

We next focus on the properties of convexity (D3) and positive homogeneity (D4) of generalized deviation measures.

\subsection{Convexity of CED}

According to \cite{FollmerSchied2002, FollmerSchied2010, FollmerSchied2011}, the  essence of diversification is encapsulated in the convexity axiom. Suppose we have two processes $X$ and $Y$ representing cumuative returns to two portfolios. Rather than investing fully in one of the two portfolios, an investor could \emph{diversify} by allocating a fraction $\lambda \in [0,1]$ of his capital to, say, $X$, and the remainder $1-\lambda$ to $Y$. Under a convex risk measure, this diversification cannot increase risk.

\begin{prop}[Convexity of CED]
Conditional Expected Drawdown is convex with respect to portfolio weights: for all $X, Y \in \Rspace^\infty$, $\lambda \in [0,1]$, and confidence level $\alpha\in(0,1)$, $\CED_\alpha(\lambda+ (1-\lambda)Y) \leq \lambda \CED_\alpha(X) + (1-\lambda)\CED_\alpha(Y)$. 
\end{prop}

\begin{proof}
For $\lambda\in[0,1]$, we have $M^{(\lambda X+(1-\lambda)Y)}\leq \lambda M^{(X)} + (1-\lambda)M^{(Y)}$ by properties of the supremum, and therefore 
\begin{eqnarray*}
D^{(\lambda X+(1-\lambda)Y)} &=& M^{(\lambda X+(1-\lambda)Y)} -\lambda X -(1+\lambda)Y \\
 &\leq& \lambda M^{(X)} + (1-\lambda)M^{(Y)}-\lambda X -(1+\lambda)Y \\
 &=& \lambda D^{(X)} + (1-\lambda)D^{(Y)}
\end{eqnarray*}
Assuming that the distributions of $\mu(X)$ and $\mu(Y)$ are continuous, and because $\mu(X)$ is defined as the supremum within the drawdown path $D$, we have $\mu(\lambda X+(1-\lambda)Y) \leq \lambda \mu(X) +(1-\lambda)\mu(Y)$. Finally, since the tail mean functional $\TM$ is subadditive and positive homogenous independent of the underlying distribution (see \citet{AcerbiTasche2002b, AcerbiTasche2002a}), and also monotonically non-decreasing, its composite with $\mu$ is also convex, and so $\CED_\alpha(\lambda+ (1-\lambda)Y) \leq \lambda \CED_\alpha(X) + (1-\lambda)\CED_\alpha(Y)$. 
\end{proof}

\begin{rem}[Drawdown risk optimization]
Convexity of CED implies that one can, in theory, allocate assets to trade off CED risk against portfolio return. There are three crucial ingredients for carrying out any optimization in practice. Convexity of the objective function to be minimized ensures that the minimum, if it exists, is a global one. The second ingredient is the feasibility and efficiency of the optimization algorithm.\footnote{\doublespacing Another crucial ingredient is having a reliable risk model feeding the optimizer with realistic and useful scenarios. This being beyond the scope of the present article, we have focused on the two main theoretical requirements in the present article. We refer the reader to \citet{Uryasev2014b}, where the theory of risk estimation and error sensitivity in the context of portfolio optimization is discussed.}  
Seminal work of \citet{RockafellarUryasev2000, RockafellarUryasev2002}, who developed an efficient linear programming (LP) algorithm for minimizing the tail mean of a distribution of returns, and of \citet{Chekhlov2003,Chekhlov2005}, who incorporated drawdown into the LP formulation, can in theory be used to minimize the tail mean of a maximum drawdown distribution. The third ingredient, which allows us to move beyond theory, is an empirically sound estimate of risk.  Further empirical exploration of the properties of CED and the study of its impact on quantitative portfolio construction, are necessary and beyond the scope of this article.\\
\end{rem}

\subsection{Positive Homogeneity of CED}
 
Degree-one positive homogenous risk measures are characterized by Euler's homogenous function theorem, and hence play a prominent role in portfolio risk analysis. More precisely, for a portfolio $P = \sum_i w_i X_i$ in $\Rspace^\infty$, a risk measure $\rho : \Rspace^\infty \to \R$ is postive homogenous of degree one if and only if $\sum_i w_i \left(\partial \rho(P)\right)/(\partial w_i) = \rho(P)$.\footnote{This formula and the topic of risk attribution is discussed in more detail in Section 4.} The risk $\rho(P)$ of the portfolio $P = \sum_i w_i X_i$ can therefore be linearly attributed along its factors $X_i$.

\begin{prop}[Positive homogeneity of CED]
Conditional Expected Drawdown is degree-one positive homogenous with respect to portfolio weights: for all $X \in \Rspace^\infty$, $\lambda > 0$ and confidence level $\alpha\in(0,1)$, $\CED_\alpha(\lambda X) = \lambda \CED_\alpha(X)$. 
\end{prop}

\begin{proof}
For $\lambda>0$, we have for $t\in[0,T]$, $M_t^{(\lambda X)} = \sup_{u\in[0,t]}(\lambda X)_u = \lambda\sup_{u\in[0,t]}(X)_u = \lambda M_t^{(X)}$, and therefore $D^{(\lambda X)} = \lambda M^{(X)} - \lambda X = \lambda D^{(X)}$. Because $\mu(X)$ is defined as the supremum within the drawdown path $D$, we have $\mu(\lambda X) = \lambda \mu(X)$. Finally, positive homogeneity of the tail mean functional yields the result.\\
\end{proof}


\section{Drawdown Risk Attribution}
\label{section:attribution}

With the theoretical framework of drawdown risk measurement in place, the next step is to understand how Conditional Expected Drawdown can be integrated in the investment process. 
We show how to systematically analyze the sources of drawdown risk within a portfolio and how these sources interact. In practice, investors may be interested in attributing risk to individual securities, asset classes, sectors, industries, currencies, or style factors of a particular risk model. In what follows, we assume a generic such risk factor model.

Fix an investment period and let $F_i$ denote the return of factor $i$ over this period ($1\leq i \leq n$). Then the portfolio return over the period is given by the sum
$$ P=\sum_{i=1}^n w_i F_i \ ,$$
where $w_i$ is the portfolio exposure to factor $i$ and the summand representing idiosyncratic risk is not included for simplicity. Because portfolio risk is not a weighted sum of source risks, there is no direct analog to this decomposition for risk measures. However, there is a parallel in terms of \emph{marginal risk contributions} (MRC), which are interpreted as a position's percent contribution to overall portfolio risk. They provide a mathematically and economically sound way of decomposing risk into additive subcomponents. 

For a risk measure $\rho$, the marginal contribution to risk of a factor is the approximate change in overall portfolio risk when increasing the factor exposure by a small amount, while keeping all other exposures fixed.\footnote{Risk contributions have become part of the standard toolkit for risk management, and they are used for risk budgeting and capital allocation. We refer the reader to \citet{Tasche2000}, \citet{Kalkbrener2005}, \citet{Denault2001}, and \citet{Qian2006} for more details.}
Formally, marginal risk contributions can be defined for any differentiable risk measure~$\rho$. 

\begin{defn}
For a factor $F_i$ in the portfolio $P=\sum_i w_i F_i$, its \emph{marginal risk contribution} $\mrc_i$ is the derivative of the underlying risk measure $\rho$ along its exposure $w_i$: 
$$\mrc_i^\rho(P) = \frac{\partial \rho(P)}{\partial w_i}.$$
If $\rho$ is homogenous of degree one, the overall portfolio risk can be decomposed using Euler's homogoneous function theorem as follows:
$$\sum_i w_i \mrc_i^\rho(P) = \sum_i \rc_i^\rho(P) = \rho(P),$$
where $\rc_i^\rho(P) = w_i \mrc_i^\rho(P)$ is the $i$-th total \emph{risk contribution} to $\rho$. Finally, \emph{fractional risk contributions} $$\frc_i^\rho(P) = \frac{\rc_i^\rho(P)}{\rho(P)}$$ denote the fractional
contribution of the $i$-th factor to portfolio risk.
\end{defn}

Risk contributions implicitly define a notion of correlation that is general enough to be defined for any risk measure. The \emph{generalized risk-based correlation} $\corr_i^\rho$ for a generic risk measure $\rho : \mathcal{M} \to \R$ between the portfolio and the $i$th asset $X_i$ is defined by: 
$$ \corr_i^\rho = \frac{\mrc_i^\rho(P)}{\rho(X_i)}  .$$
Generalized correlations are monotonically decreasing in position weight. Factoring out the $i$th marginal risk $\rho(X_i)$ from the $i$th risk contribution $\rc_i(P)$, we obtain the generalized form of the ``X-Sigma-Rho" decomposition of \citet{MencheroPoduri2008}:
$$ \rc^\rho_i(P) =  w_i \rho(X_i) \frac{\mrc^\rho_i(P)}{\rho(X_i)} = w_i \rho(X_i) \corr^\rho_i \quad .$$
We refer the reader \citet{GoldbergMencheroHayesMitra2010} for a more detailed development of generalized correlations.

\subsection{Drawdown Risk Contributions}

\citet{MencheroPoduri2008} and \citet{GoldbergMencheroHayesMitra2010} developed a standard toolkit for analyzing portfolio risk using a framework centered around marginal risk contributions. By integrating drawdown risk into this framework, investors can estimate how a trade would impact the overall drawdown risk of the portfolio. Because  Conditional Expected Drawdown is positive homogenous, the individual factor contributions to drawdown risk add up to the overall drawdown risk within a path $P\in\Rspace^\infty$ of returns to a portfolio with values at time $t\leq T$ given by $P_t = \sum_i w_i F_{i,t}$\footnote{The process corresponding to the $i$-th factor is written $F_i$, and its instance at time $t\in[0,T]$ is denoted by $F_{i,t}$.}:
\begin{equation}
\CED_\alpha(P) = \sum_i {w_i {\mrc}_i^{\CED_\alpha}(P)} , \quad \alpha \in [0,1].
\label{eqn:ced_decomp_mrc}
\end{equation}

Recall that a marginal risk contribution is a partial derivative, and so practitioners can implement Formula~\ref{eqn:ced_decomp_mrc} using numerical differentiation.  However, this tends to introduce noise. We next show that an individual marginal contribution to drawdown risk can be expressed as an integral, and this reduces noise, since integration is a smoothing operator.\footnote{This is analogous to marginal contributions to Expected Shortfall, which can also be expressed as integrals; see \cite{Tasche2000} and \cite{Tasche2002} where it is shown that for quantile based risk measures (such as VaR and ES, but also spectral measures), an Euler attribution can be expressed as an intuitive expectation.}
Indeed, the individual marginal contribution ${\mrc}_i^{\CED_\alpha}$ of the $i$-th factor to overall portfolio drawdown risk $\CED_\alpha(P)$ is given by the expected drop of the $i$-th factor in the interval $[s^*,t^*] \subset[0,T]$ where the overall portfolio maximum drawdown $\mu(P)$ occurs, given that the maximum drawdown of the overall portfolio exceeds the drawdown threshold. This definition is analogous to the marginal contribution to shortfall, and we formalize it next.

\begin{prop}
Marginal contributions to drawdown risk are given by:
\begin{equation}
 \mrc_i^{\CED_\alpha} (P) = \E \left[ \left(F_{i,t^*} - F_{i,s^*} \right) \mid \mu(P)  > \mathrm{DT}_\alpha(P) \right] ,
\label{eqn:ced_mrc}
\end{equation}
where $\CED_\alpha(P)$ is the overall portfolio CED, $\mu(P)$ is the maximum drawown random variable, $\rm{DT}_\alpha(P)$ is the portfolio maximum drawdown threshold at $\alpha$, and $s^*<t^*\leq T$ are random times such that:
$$ \mu(P) = P_{t^*} - P_{s^*}, $$  
and we assume that the maximum drawdown of $P= \sum_i w_i F_{i}$ is strictly positive.
\label{prop:ced_mrc}
\end{prop}

\begin{proof}
We use the results of \cite{Tasche2002}, \citet{GoldbergMencheroHayesMitra2010} and \citet{QRM2005}, who show that the $i$-th marginal contribution to Expected Shortfall $\ES_\alpha$ at confidence level $\alpha \in (0,1)$ of a random variable $L = \sum_i w_i Y_i$ representing portfolio loss is given by
\begin{equation}
\mrc_i^{\ES_\alpha} (L) = \E \left[ Y_i \mid L > \Var_\alpha(L)\right] ,  
\label{eqn:mrc_es}
\end{equation}
where $\Var_\alpha(L)$ denotes the Value-at-Risk of $L$ at $\alpha$, that is the $\alpha$-quantile of the loss distribution $L$. 

We derive an analog to Formula \ref{eqn:mrc_es}.
Assuming that the maximum drawdown of $P = \sum_i w_i F_{i}$ is strictly positive, let 
$$ \mu(P) =  P_{t^*} - P_{s^*} $$  
for some $s^*<t^*\leq T$. Then the $i$-th marginal contribution $ \mrc_i^{\CED_\alpha} (P) $ to overall portfolio drawdown risk $\CED_\alpha (P)$ is given by

\begin{eqnarray}
 \mrc_i^{\CED_\alpha} (P) &=& \frac{\partial }{\partial w_i} \left( \TM_\alpha \left( \mu(P) \right) \right)  \nonumber \\
&=& \frac{\partial }{\partial w_i}  \E \left[\mu(P) \mid \mu(P)  > \mathrm{DT}_\alpha(P)\right] \nonumber \\
&=& \frac{\partial }{\partial w_i}  \E \left[ (P_{t^*} - P_{s^*}) \mid \mu(P)  > \mathrm{DT}_\alpha(P)\right] \nonumber \\
&=& \frac{\partial }{\partial w_i}  \E \left[ \left(\sum_{i=1}^n w_i F_{i,t^*} - \sum_{i=1}^n w_i F_{i,s^*}  \right) \mid \mu(P)  > \mathrm{DT}_\alpha(P)\right] \nonumber \\
&=& \frac{\partial }{\partial w_i} \E \left[ \sum_{i=1}^n w_i \left( F_{i,t^*} - F_{i,s^*} \right) \mid \mu(P)  > \mathrm{DT}_\alpha(P) \right] \nonumber \\
\label{eqn:mrc_ced_derivation}
&=& \frac{\partial }{\partial w_i}  \left( \sum_{i=1}^n w_i \E \left[ \left(F_{i,t^*} - F_{i,s^*} \right) \mid\mu(P)  > \mathrm{DT}_\alpha(P) \right] \right) \\ \nonumber
 \end{eqnarray}

Using the fact that the partial derivative with respect to a quantile is zero, as discussed by \citet{Bertsimas2004}, Formula \ref{eqn:mrc_ced_derivation} simplifies to:
$$ \mrc_i^{\CED_\alpha} (P) = \E \left[ \left(F_{i,t^*} - F_{i,s^*} \right) \mid \mu(P)  > \mathrm{DT}_\alpha(P) \right] .$$

Finally, note that the variables $s^*$ and $t^*$ are stochastic. This means that in a Monte Carlo simulation of a discretized version of this problem, they will take on a different value scenario by scenario. 
\end{proof}


\section{Empirical Analysis of Drawdown Risk}
\label{section:empirical}

We analyze historical values of Conditional Expected Drawdown based on daily data for two asset classes: US Equity and US Government Bonds. The US Government Bond Index we use\footnote{See Appendix A for details on the data and their source.} includes fixed income securities issued by the US Treasury (excluding inflation-protected bonds) and US government agencies and instrumentalities, as well as corporate or dollar-denominated foreign debt guaranteed by the US government, with maturities greater than 10 years. These include government agencies such as the Federal National Mortgage Association (Fannie Mae) and the Federal Home Loan Mortgage Corporation (Freddie Mac) without an explicit guarantee. In comparison to US \emph{Treasury} Bond Indices, US Government Bond Indices were highly volatile and correlated with US Equities during the financial crisis of 2008. The effect of this will be seen in our empirical analysis.\footnote{We thank Robert Anderson for pointing out the important distinction between US Government Bond and US Treasury Bond Indices.}
Summary risk statistics for the two asset classes and three fixed-mix portfolios are shown in Table \ref{stats}.

\subsection{Time-varying Drawdown Risk Concentrations}

\begin{table}[t]
\centering 
\begin{tabular}{l || r r r r r } 
\hline 
 & Volatility & ES$_{0.9}$ & CED$_{0.9}$ (6M-paths) & CED$_{0.9}$ (1Y-paths) & CED$_{0.9}$ (5Y-paths)\\ [0.5ex] 
\hline \hline
US Equity & 18.35\%  & 2.19\% & 47\% & 51\% & 57\% \\ 
US Bonds &  5.43\%  & 0.49\% & 29\% & 32\% & 35\% \\
50/50      &   9.53\%  & 1.30\% & 31\% & 32\% & 35\% \\
60/40      &   11.12\%  & 1.35\% & 33\% & 35\% & 38\% \\
70/30      &  12.92\% & 1.40\% & 36\% & 40\% & 44\% \\
\hline
\end{tabular}
\vspace{20pt}
\caption{\onehalfspacing Summary statistics for daily US Equity and US Bond Indices and three fixed-mix portfolios over the period 1 January 1982 to 31 December 2013. Expected Shortfall and Conditional Expected Drawdown are calculated at the 90\% confidence level. Three drawdown risk metrics are calculated by considering the maximum drawdown within return paths of different fixed lengths (6 months, 1 year and 5 years).} 
\label{stats} 
\end{table}

Using the definition of marginal contributions to Conditional Expected Drawdown (derived in Proposition \ref{prop:ced_mrc}), we look at the time varying contributions to CED.
Figure \ref{frc_dd} displays the daily 6-month rolling fractional contributions to drawdown risk $\CED_{0.9}$ (at the 90\% threshold of the 6-month maximum drawdown distribution) of the two asset classes (US Equity and US Bonds) in the balanced 60/40 allocation.\footnote{See Appendix B for details on the  risk estimation and portfolio construction methodologies used. Note also that similar effects can be seen in other fixed-mix portfolios, such as the equal-weighted 50/50 portfolio and the 70/30 allocation. In the following empirical analyses, we will be focusing exclusively on the traditional 60/40 allocation.} Between 1982 and 2008, and between 2012 and 2013, the contributions of US Equity to overall drawdown risk fluctuated between 80\% and 100\%. Note that this period includes two of the three turbulent market regimes that occurred during this 30-year window, namely the 1987 stock market crash and the burst of the internet bubble in the early millennium. During the credit crisis of 2008, however, we see, unexpectedly, that bonds contributed almost as much as equities to portfolio drawdown risk. 

\begin{figure}
\centering
  \includegraphics[width=0.7\linewidth]{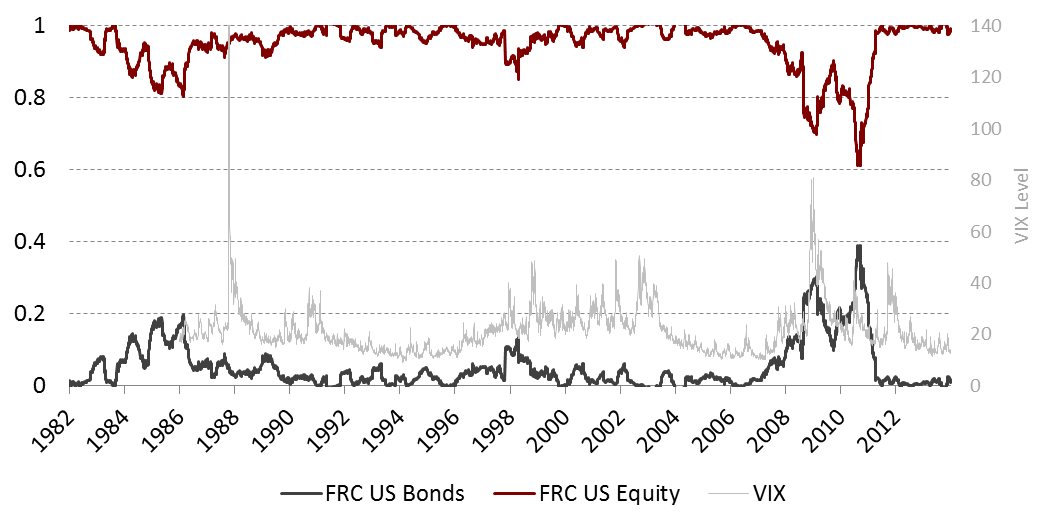}
\caption{\onehalfspacing Daily 6-month rolling Fractional Risk Contributions (FRC) along the 90\% Conditional Expected Drawdown (CED) of US Equity and US Bonds to the balanced 60/40 portfolio. Also displayed is the daily VIX series over the period 1982 -- 2013, with the right-hand axis indicating its level.}
\label{frc_dd}
\end{figure}

\begin{figure}
\centering
  \includegraphics[width=0.95\linewidth]{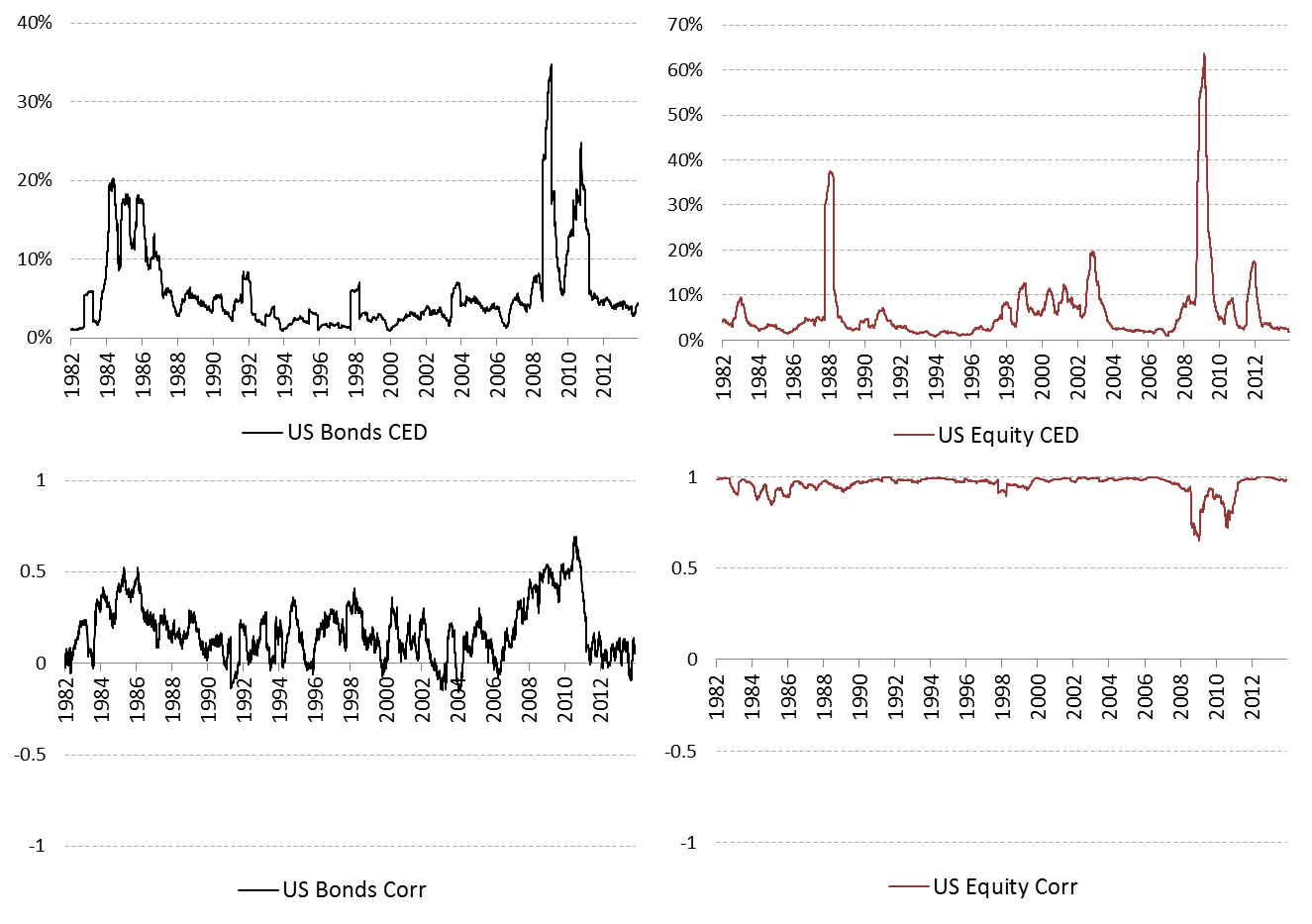}
\caption{\onehalfspacing Decomposition of the individual contributions to drawdown risk $ \rc^\CED_i(P) = w_i \CED(X_i) \corr^\CED_i$ for the 60/40 allocation to US Equity and US Bonds over the period 1982 -- 2013. The top two panels show the daily 6-month rolling standalone 90\% Conditional Expected Drawdown (CED) of the two asset classes, while the bottom two panels show the daily 6-month rolling generalized correlations of the individual assets along CED.}
\label{rc_decomp}
\end{figure}

Our analysis shows little connection between market turbulence and drawdown risk concentration in  the 60/40 fixed mix of US Equity and US Bonds.  Notably, the most equitable attribution of drawdown risk occurred during the 2008 financial crisis. This can be explained by the inclusion of bonds issued by  Fannie Mae and Freddie Mac in the  US Government Bond Index. In calm regimes, these Agency Bonds tended to be  correlated with US Treasury bonds, but during the financial crisis, Agency Bonds were more correlated with US Equity.  For comparison, we provide the same analysis when the underlying Bond Index used is the US Treasury Bond Index (see Figures \ref{frc_dd_2} and \ref{rc_decomp_2} in Appendix B). In this case, as one would expect, the least equitable attribution of drawdown risk occurred during turbulent market periods.

 To understand the sources of the risk contributions, particularly during the credit crisis of 2008 where the concentrations of US Equity and US Government Bonds approached parity, we carry out the ``X-Sigma-Rho" decomposition of \citet{MencheroPoduri2008}. Recall from Section 4 that risk contribution is proportional to the product of standalone risk and  generalized correlation. In the case of Conditional Expected Drawdown, this means that:
$$ \rc^\CED_i(P) = w_i \CED(X_i) \corr^\CED_i .$$

Because we are working with a fixed-mix portfolio, the exposures $w_i$ are constant: 0.6 and 0.4 for US Equity and US Bonds, respectively. This means that the time-varying risk contributions of Figure \ref{frc_dd} depend on the time-varying drawdowns ($\CED(X_i)$) and correlations ($\corr^\CED_i$). Figure \ref{rc_decomp} displays these for each of the two assets in our 60/40 portfolio. Observe that during the 2008 financial crisis, both the drawdown risk contribution of US Bonds and its generalized correlation  were elevated relative to the subsequent period. On the other hand, the generalized correlation of US Equity during the 2008 crisis decreased.  The combination of these effects may have driven the changes in the drawdown contributions of US Bonds and US Equity during the 2008 crisis.\footnote{For comparison, we include in Figure \ref{es_decomp} of Appendix C the risk decomposition along Expected Shortfall.}

In Section \ref{section:serial}, we give a statistical analysis that supports the economic explanation of the increased CED values for US Government Bonds. In practice, investors can efficiently control such regime-dependent fluctuations in drawdown risk concentrations since Conditional Expected Drawdown is a convex risk measure; that is both the return path and the drawdown path are convex with respect to asset weights. Hence, they are convex functions of factors that are linear combinations of asset weights. This implies that reducing the portfolio exposure to an asset or factor in a linear factor model decreases its marginal contribution to overall portfolio drawdown.

It is possible for a portfolio to have equal risk contributions with respect to one measure while harboring a substantial concentration with respect to another.\footnote{Risk parity portfolios, which are constructed to equalize risk contributions, have been  popular investment vehicles in the wake of the 2008 financial crisis (see \citet{AndersonBianchiGoldberg2012} and \citet{AndersonBianchiGoldberg2013}).  This is in spite of the fact that there may be no theoretical basis for the construction. } Figure \ref{frc_parity} illustrates such a case. Four portfolios are constructed to be maximally diversified along the following risk measures: volatility, Expected Shortfall, and Conditional Expected Drawdown. The underlying asset classes are US Equity and US Government Bonds as before from 1982 to 2013.\footnote{See Appendix B for details on the data, risk estimation, and portfolio construction methodologies used.} We refer to these as being in \emph{parity} with respect to the underlying risk measure. The confidence level for both ES and CED is fixed at 90\%. Figure \ref{frc_parity} shows fractional risk contribution of the equity component to each of three risk measures in three types of risk parity portfolios. Concentrations in terms of drawdown risk, in particular, are revealed. For instance, even though the ES Parity portfolio, which has equal contributions to Expected Shortfall, is constructed to minimize downside risk concentrations, it turns out to have 75\% of its drawdown risk concentrated in US Equity.

\begin{figure}
\centering
  \includegraphics[width=0.7\linewidth]{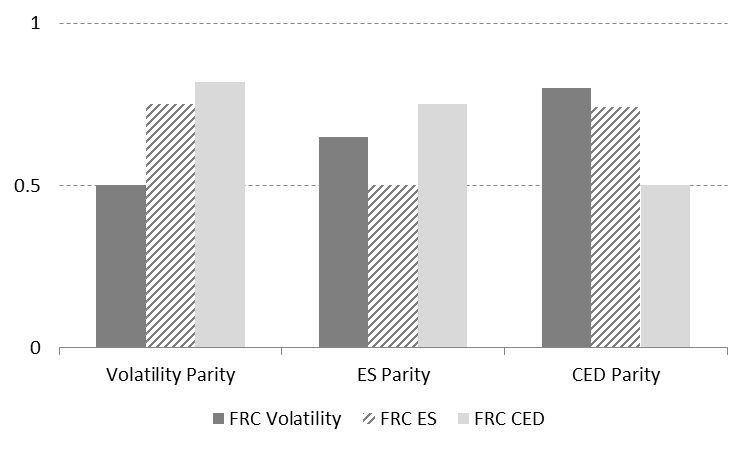}
\caption{\onehalfspacing Fractional Risk Contributions (FRC) of US Equity measured along three different risk measures (volatility, 90\% Expected Shortfall and 90\% Conditional Expected Drawdown) for the following two-asset portfolios consisting of US Equity and US Bonds over the period 1982--2013: Volatility Parity, ES Parity and CED Parity. Each parity portfolio is constructed to have equal risk contributions along its eponymous risk measure.}
\label{frc_parity}
\end{figure}


\subsection{Drawdown Risk and Serial Correlation}\label{section:serial}

One advantage of looking at maximum drawdown distributions rather than return distributions, and thus Conditional Expected Drawdown rather than Expected Shortfall, lies in the fact that drawdown is inherently path dependent. In other words, drawdown measures the degree to which losses are sustained, as small but persistent cumulative losses may still lead to large drops in portfolio net asset value, and hence may force liquidation. On the other hand, volatility and Expected Shortfall fail to distinguish between intermittent and consecutive losses. We show that, to a greater degree than these two risk measures, Conditional Expected Drawdown captures temporal dependence. Moreover, the effect of serial correlation on drawdown risk can  be seen in the drawdown risk contributions. \\

\textbf{An increase in serial correlation increases drawdown risk.}
To see how temporal dependence affects risk measures, we use Monte Carlo simulation to generate an autoregressive AR(1) model:
$$ r_t = \kappa r_{t-1}  + \epsilon_t ,$$
with varying values for the autoregressive parameter $\kappa$ (while $\epsilon$ is Gaussian with variance 0.01), and calculate volatility, Expected Shortfall, and Conditional Expected Drawdown of each simulated autoregressive time series. Figure \ref{kappa} displays the results. All three risk measures were affected by the increase in the value of the autoregressive parameter, but the increase is steepest by far for CED. We next use maximum likelihood to fit the AR(1) model to the daily time series of US Equity and US Government Bonds on a 6-month rolling basis to obtain time series of estimated $\kappa$ values for each asset. The correlations of the time series of $\kappa$  with the time series of 6-month rolling volatility, Expected Shortfall, and Conditional Expected Drawdown are shown in Table \ref{kappa_corr}. The correlations are substantially higher for US Equity across all three risk measures. Note that for both asset classes, the correlation with the autoregressive parameter is highest for CED. Figure \ref{kappa_ced} contains the scatter plots of estimated $\kappa$ parameters for US Equity and US Bonds against their CED. \\

\begin{figure}
\centering
  \includegraphics[width=0.6\linewidth]{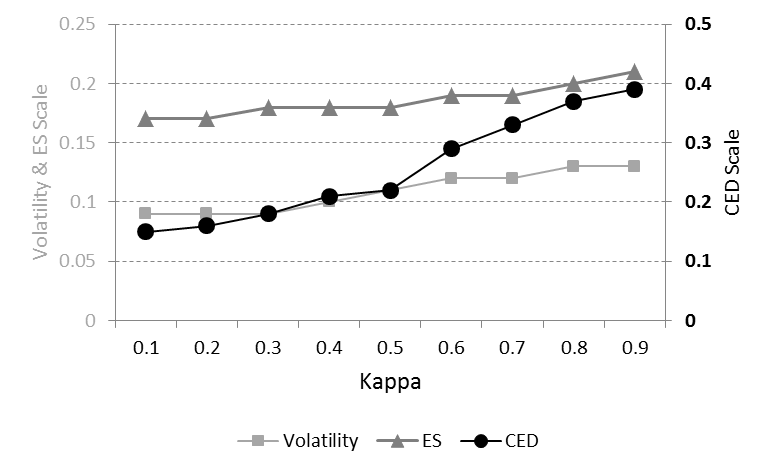}
\caption{\onehalfspacing Volatility, 90\% Expected Shortfall (ES), and 90\% Conditional Expected Drawdown (CED) of a Monte Carlo simulated AR(1) model (with 10,000 data points) for varying values of the autoregressive parameter $\kappa$.}
\label{kappa}
\end{figure}

\begin{table}
\centering 
\begin{tabular}{l || r r r } 
\hline 
 & Volatility & ES$_{0.9}$ & CED$_{0.9}$ \\ [0.5ex]
\hline \hline
US Equity & 0.47  & 0.52 & 0.75 \\ 
US Bonds &  0.32  & 0.39 & 0.69  \\
\hline 
\end{tabular}
\vspace{20pt}
\caption{\onehalfspacing For the daily time series of each of US Equity and US Government Bonds, correlations of estimates of the autoregressive parameter $\kappa$ in an AR(1) model with the values of the three risk measures (volatility, 90\% Expected Shortfall and 90\% Conditional Expected Drawdown) estimated over the entire period (1982--2013).} 
\label{kappa_corr} 
\end{table}

\begin{figure}
\centering
  \includegraphics[width=\linewidth]{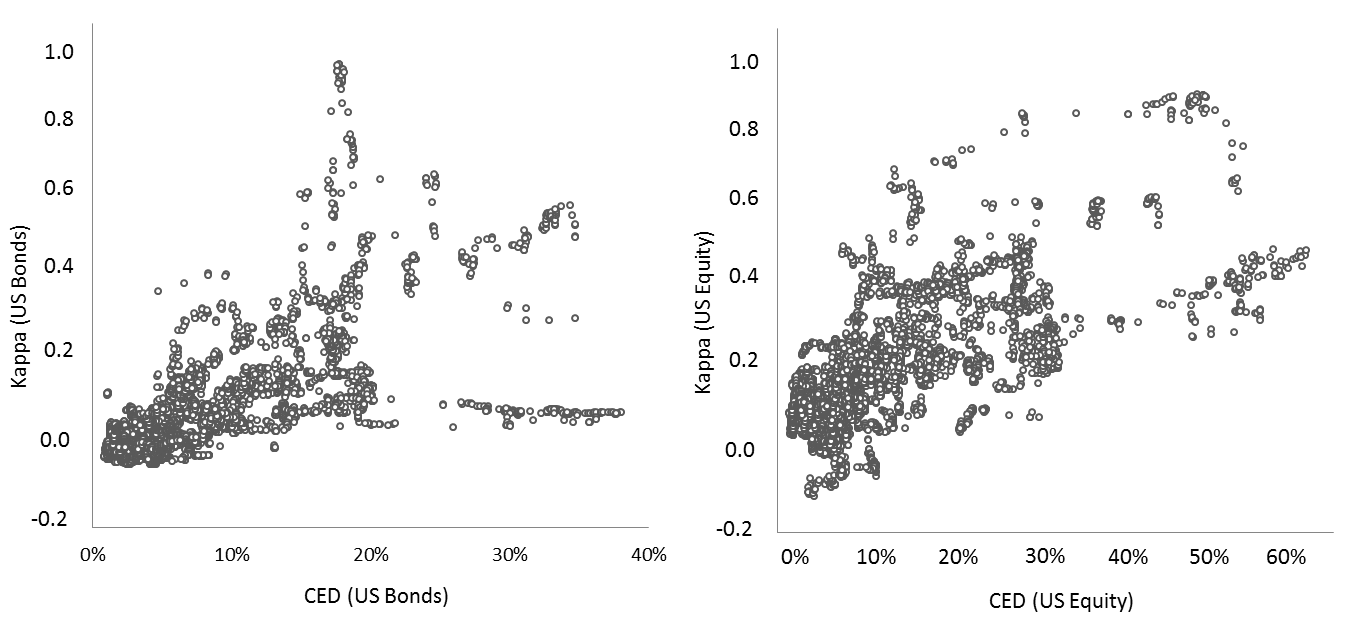}
\caption{\onehalfspacing For each of US Equity and US Government Bonds, scatter plots of the daily time series of 6-month rolling estimates of the autoregressive parameter $\kappa$ with the 6-month rolling estimates of 90\% Conditional Expected Drawdown.}
\label{kappa_ced}
\end{figure}

\textbf{An increase in serial correlation increases drawdown risk concentrations.}
We now show how temporal dependence is manifest in the drawdown risk contributions. Figure \ref{frc_data} shows the fractional risk contributions over the entire period 1982--2013 of US Equity to the balanced 60/40 portfolio for three risk measures, volatility, ES, and CED, based on daily data. The fractional contributions of  US Equity to volatility and ES were large (over 90\%) and close in magnitude. For CED, however, the concentration was less pronounced, which means that the  contribution of US Bonds to drawdown risk exceeded its contribution to volatility and shortfall risk.  A candidate explanation is temporal dependence: while bonds systematically have lower volatility and shortfall risk than do equities, they do occasionally suffer from extended periods of consecutive losses.

To test this hypothesis, we simulate the returns $r_E$ and $r_B$ to two assets $E$ and $B$ representing equities and bonds, respectively, with an autoregressive AR(1) model:
$$ r_{E,t} = \kappa_E r_{E,t-1}  + \epsilon_{E,t} ,$$
and 
$$ r_{B,t} = \kappa_B r_{B,t-1}  + \epsilon_{B,t} ,$$
and we construct a simulated 60/40 fixed-mix portfolio.
The  AR(1) model parameters are obtained by calibrating  to daily time series of US Equity and US Bonds. The estimated autoregressive parameters are $\kappa_E = 0.43$ and $\kappa_B = 0.35$. We assume the  $\epsilon$ variable is  Gaussian, with volatility of 18.4\% for asset E (based on the volatility of US Equity)  and 5.5\% for asset B, (based on the volatility of US Bonds). From the simulated data, we fit AR(1) models and their fractional contributions to  volatility, ES and CED. When using only the residuals, we obtain statistically equal risk contributions since the innovations are Gaussian. However, without removing the autoregressive component,  contributions to CED once again differ from contributions to volatility and ES. Figure \ref{frc_sim} displays the corresponding fractional contributions of the more volatile asset class, $E$,   to the three risk measures. Note that the two panels in Figure \ref{frc_comparison} are visually indistinguishable even though one is based on  historical data, whereas the other is simulated.

\begin{figure}
\centering
\begin{subfigure}{.5\textwidth}
  \centering
  \includegraphics[width=\linewidth]{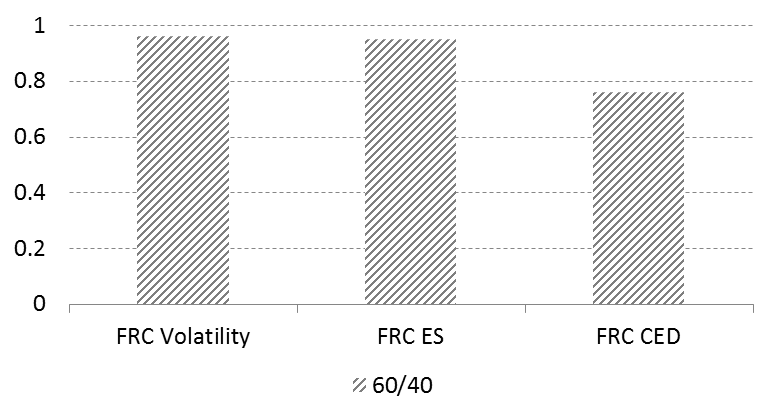}
  \caption{}
  \label{frc_data}
\end{subfigure}%
\begin{subfigure}{.5\textwidth}
  \centering
  \includegraphics[width=\linewidth]{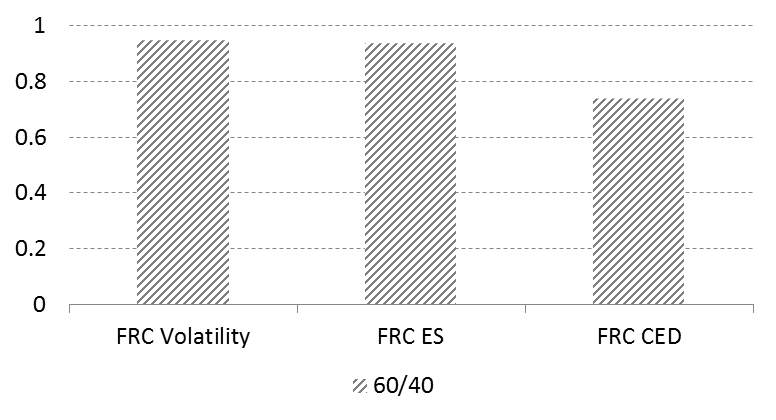}
  \caption{}
  \label{frc_sim}
\end{subfigure}
\caption{\onehalfspacing (A) Fractional contributions over the entire period 1982--2013 of the US Equity asset to volatility, 90\% Expected Shortfall and 90\% Conditional Expected Drawdown in the 60/40 portfolio, based on daily data. (B) Fractional contributions of the simulated high-volatility AR(1) asset to volatility, 90\% Expected Shortfall and 90\% Conditional Expected Drawdown in the 60/40 portfolio.}
\label{frc_comparison}
\end{figure}

\section{Drawdown: From Practice to Theory and Back Again}
\label{section:conclusion}

Financial practitioners  rely on maximum drawdown as an indicator of investment risk. However, due to its inherent path dependency, maximum drawdown has tended to fall outside of probabilistic treatments of investment risk, which focus on return and loss distributions at fixed horizons.  As a result, maximum drawdown has been excluded from  standard portfolio analysis toolkits that attribute risk to factors or asset classes, and that use risk forecasts as counterweights to expected return in portfolio construction routines.

In this article, we develop a new probabilistic measure of drawdown risk, Conditional Expected Drawdown (CED), which is the tail-mean of a drawdown distribution at a fixed horizon.  Since CED is perfectly analogous to the familiar return-based risk measure, Expected Shortfall (ES),  CED is easy for practitioners to interpret and it enjoys desirable theoretical properties of tail-means such as positive degree-one homogeneity and convexity. Thus, the development of a consistent theory for drawdown facilitates an extension of its current practical applications.

The path dependency of  Conditional Expected Drawdown makes it more sensitive to serial correlation than Expected Shortfall or volatility.  We demonstrate this using a simulated AR(1) model. All else equal,   CED increases much more rapidly as a function of the autoregressive parameter $\kappa$ than do Expected Shortfall or volatility.  In an empirical study, we find relatively high correlations between serial correlation and estimated CED (.75 for US Equity, .69 for US Bonds) compared to Expected Shortfall (.52 for US Equity, .39 for US Bonds)  and volatility (.47 for US Equity, .32 for US Bonds).

Since it  is positive degree-one homogenous, CED (like ES and volatility) can be decomposed into a sum of risk contributions, and the relative sensitivity of CED to serial correlation is manifest in risk concentrations. In an empirical study of a balanced 60/40 portfolio of US Equity and US Bonds over the period 1982--2013, US Equity accounted for roughly  75\% of CED, but more than 90\% of ES and volatility. A plausible explanation is the relatively high level of serial correlation in US Bonds.   We support this hypothesis with another simulation:  we replicate the empirically observed  concentrations of CED, ES and volatility using a simulated 60/40 balanced portfolio based on  AR(1) models calibrate to US Equity and US Bonds over the study period.

Since CED is convex, it can serve as a counterweight to expected return in a quantitative optimization.  Exploiting the parallels between Expected Shortfall as a tail-mean of a return distribution and Conditional Expected Drawdown as a tail-mean of a drawdown distribution, one can in theory use the linear programming algorithm developed by \cite{RockafellarUryasev2000, RockafellarUryasev2002}.  

This article lays the foundation needed to incorporate  Conditional Expected Drawdown in the investment process.  Further empirical exploration of the properties of CED,  research into the incremental information it adds beyond what is in return-based risk measures, and the study of its impact on quantitative portfolio construction, are the next steps.

\newpage


\appendix

\section{Data and Estimation Methodologies}

\subsection{Data}

The data were obtained from the Global Financial Data database. We took the daily time series for the S\&P 500 Index and the USA 10-year Government Bond Total Return Index.

\subsection{Portfolio Construction}

Rather than provide thorough realistic empirical analyses of portfolio risk and return, our goal behind the simulated portfolios is to illustrate this article's theoretical development in relation to drawdown risk. For simplicity, we therefore do not account for transaction costs or market frictions in all hypothetical portfolios constructed throughout this study. Moreover, we assume that all portfolios are fully invested and long only. \\

\textbf{Fixed-mix portfolios.}
In the fixed-mix portfolios, rebalancing to the fixed weights is done on a monthly basis. When comparing to other popular rebalancing schemes (quarterly, bi-annually and yearly), similar results were obtained.  \\

\textbf{Risk parity portfolios.} In risk parity strategies, assets are weighted so their ex post risk contributions are equal. As mentioned in Section 5, parity portfolios are not restricted to volatility only, but can be constructed along other risk measures, such as Expected Shortfall and Conditional Expected Drawdown. Asset weights in the strategies depend on estimates of the underlying risk measures (see Section A.3), which are calculated using a 3-year rolling window of trailing returns. Varying the estimation methodology by changing the length of the rolling window or the weighting scheme applied to the returns within this window did not alter our results substantially. Similar to the fixed-mix portfolios, risk parity portfolios are rebalanced monthly, with other rebelancing schemes yielding similar results.

\subsection{Risk Estimation}

${}$ \\

\textbf{Volatility.} Portfolio volatility is calculated as the annualized standard deviation of the daily time series over the entire period under consideration. To obtain the volatility risk contributions for a $n$-asset portfolio $P = \sum_i w_i X_i$, note that the $i$-th total contribution $\rc_i^\sigma$ to portfolio volatility 
$$\sigma(P) = \sum_i w_i^2 \sigma_i^2 + \sum_i \sum_{j \neq i} w_i w_j \sigma_{i,j}$$
 is 
$$\rc_i^\sigma = w_i^2 \sigma_i^2 + \sum_{j \neq i} w_i w_j \sigma_{i,j},$$ 
where $\sigma_i^2$ is the variance of $X_i$ and $\sigma_{i,j}$ is the covariance of $X_i$ and $X_j$. Then, the $i$-th fractional contribution to volatility is given by
$$  \frc_i^\sigma (P) = \frac{w_i^2 \sigma_i^2 + \sum_{j \neq i} w_i w_j \sigma_{i,j}}{\sigma(P)} \quad .$$

\textbf{Expected Shortfall.}
For confidence level $\alpha \in (0,1)$, an estimate for the Expected Shortfall of a portfolio is calculated by ordering the daily return time series over the whole period according to the magnitude of the returns, then averaging over the worst $(1-\alpha)$ percent outcomes, more specifically:
$$ \widehat{\ES}_\alpha = \frac{1}{K} \sum_{i=1}^{K} r_{(i)}  ,  $$
where $T$ is the length of the daily time series, $K = \lfloor T(1-\alpha) \rfloor$,  and $r_{(i)}$ is the $i$-th return of the \emph{magnitude-ordered} time series. To obtain the contributions to shortfall risk, recall that under a continuity assumption, the Expected Shortfall of an asset $X \in \mathcal{M}$ can be expressed as $\ES_\alpha (X) = \E\left(X \mid X \geq \VAR_\alpha(X)\right)$, or the expected loss in the event that its Value-at-Risk at $\alpha$ is exceeded.\footnote{See for example \citet{QRM2005}.} As usual, let $P = \sum_i w_i X_i$ be the portfolio in consideration. Assuming differentiability of the risk measure VaR, the marginal contribution of $X_i$ to portfolio shortfall $\ES_\alpha (P)$ is given by
$$ \mrc_i^{\ES_\alpha}(P) =\frac{\partial \ES_\alpha(P)}{\partial w_i} =  \E(X_i \mid P \geq \VAR_\alpha(P)) \quad.$$
An estimate for the $i$-th marginal contribution to shortfall risk is then obtained by averaging over all the returns of asset $X_i$ that coincide with portfolio returns exceeding the portfolio's Value-at-Risk at threshold $\alpha$.\\

\textbf{Conditional Expected Drawdown.}
The first step in calculating an estimate for Conditional Expected Drawdown is to obtain the empirical maximum drawdown distribution. From the historical time series of returns, we generate return paths of fixed length $n$ using a one-day rolling window. This means that consecutive paths overlap. The advantage is that for a return time series of length $T$, we obtain a maximum drawdown series of length $T-n$, which for large $T$ and small $n$ is fairly large, too. From these $T-n$ return paths we calculate the maximum drawdown as defined in Section 2. An estimate for the Conditional Expected Drawdown at confidence level $\alpha \in (0,1)$ is then calculated as the average of the largest $(1-\alpha)$ percent maximum drawdowns.  To obtain an estimate for the $i$-th contribution to drawdown risk CED, we take the average over all the drawdowns of the $i$-th asset in the path $[t_{j*},t_{k*}]$ that coincide with the overall portfolio's maximum drawdowns that exceed the portfolio's drawdown threshold $\textrm{DT}_\alpha$ at confidence level $\alpha$. (Recall that $j^*<k^*\leq n$ are such that 
$ \boldsymbol\mu(P_{T_n}) =  P_{t_{k*}} - P_{t_{j*}}. $)

\pagebreak

\section{Drawdown risk decomposition along a balanced portfolio of US Equity and US Treasury Bonds}


\begin{figure}[h]
\centering
\includegraphics[width=0.5\linewidth]{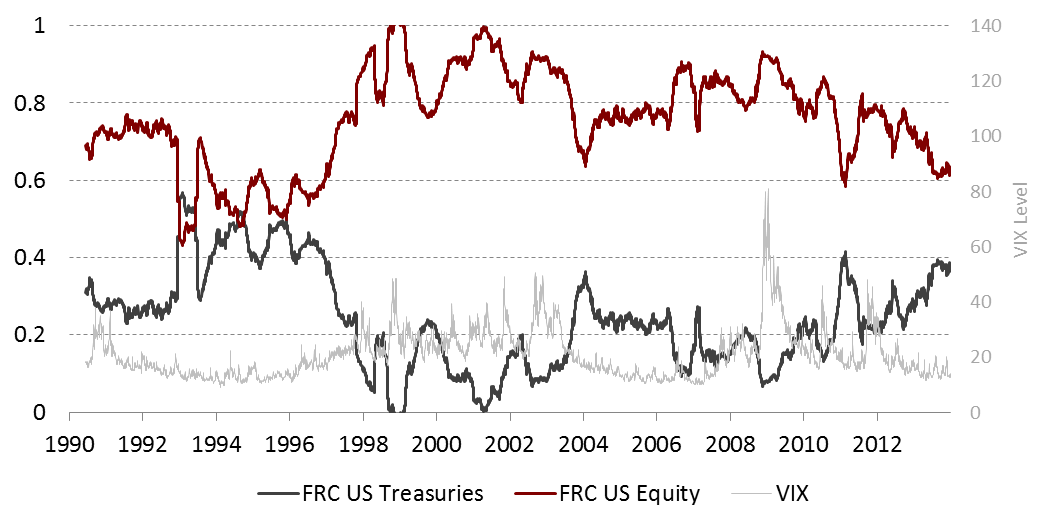}
\caption{\onehalfspacing Daily 6-month rolling Fractional Risk Contributions (FRC) along 90\% Conditional Expected Drawdown (CED) of US Equity and US Treasury Bonds to the balanced 60/40 portfolio over the period 1982--2013. Also displayed is the daily VIX series over the same period, with the right-hand axis indicating its level.}
\label{frc_dd_2}
\end{figure}

\begin{figure}[h]
\centering
\includegraphics[width=0.7\linewidth]{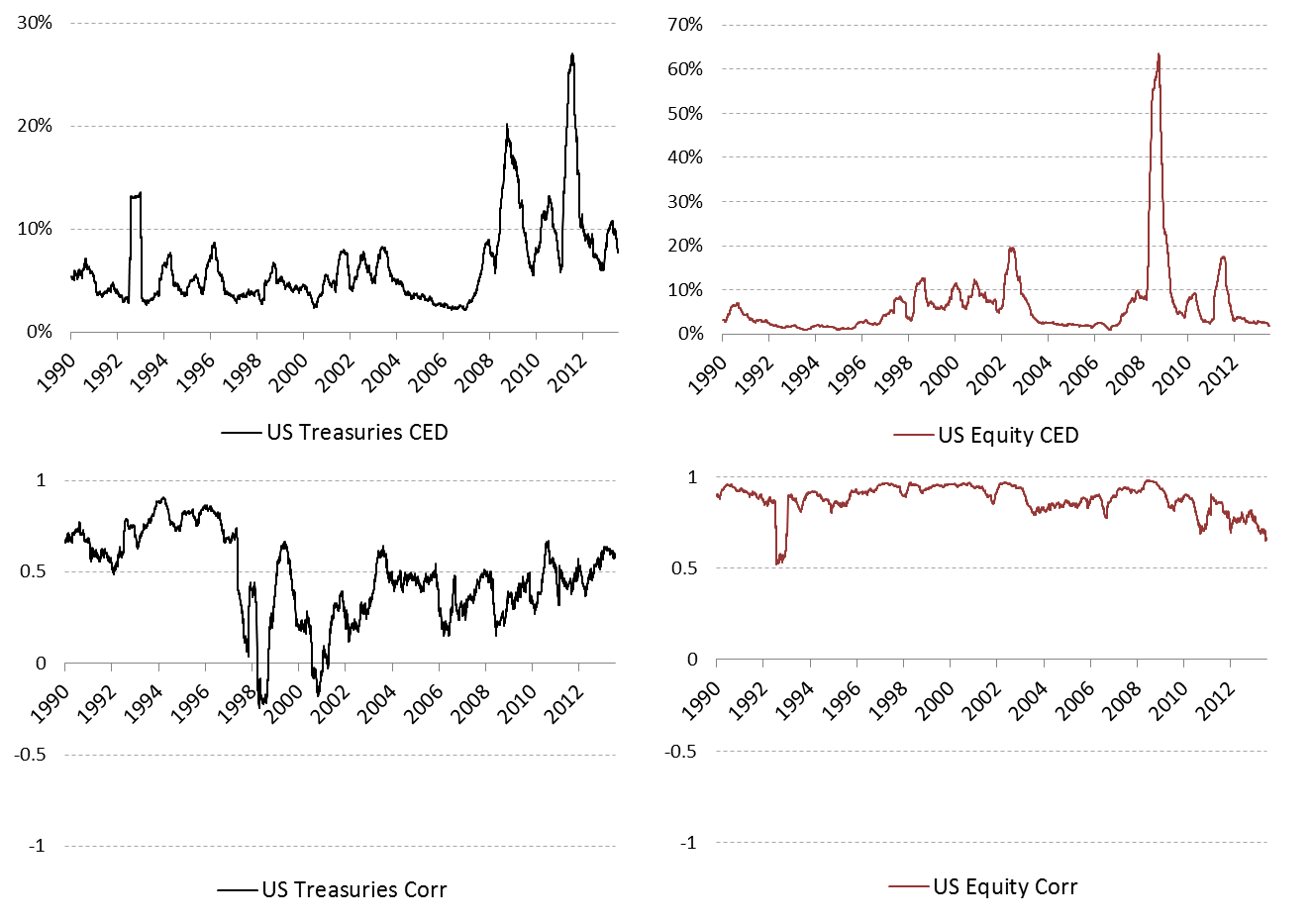}
\caption{\onehalfspacing Decomposition of drawdown risk contributions $ \rc^\CED_i(P) = w_i \CED(X_i) \corr^\CED_i$ for the 60/40 allocation to US Equity and US Treasury Bonds over the period 1982--2013. The top two panels show the daily 6-month rolling standalone 90\% Conditional Expected Drawdown (CED) of the two assets, while the bottom two panels show the 6-month rolling generalized correlations of the individual assets along CED.}
\label{rc_decomp_2}
\end{figure}

\pagebreak

\section{Risk decomposition along Expected Shortfall}

\vspace{30pt}

\begin{figure}[h]
\centering
\begin{subfigure}{\textwidth}
\centering
\includegraphics[width=0.7\linewidth]{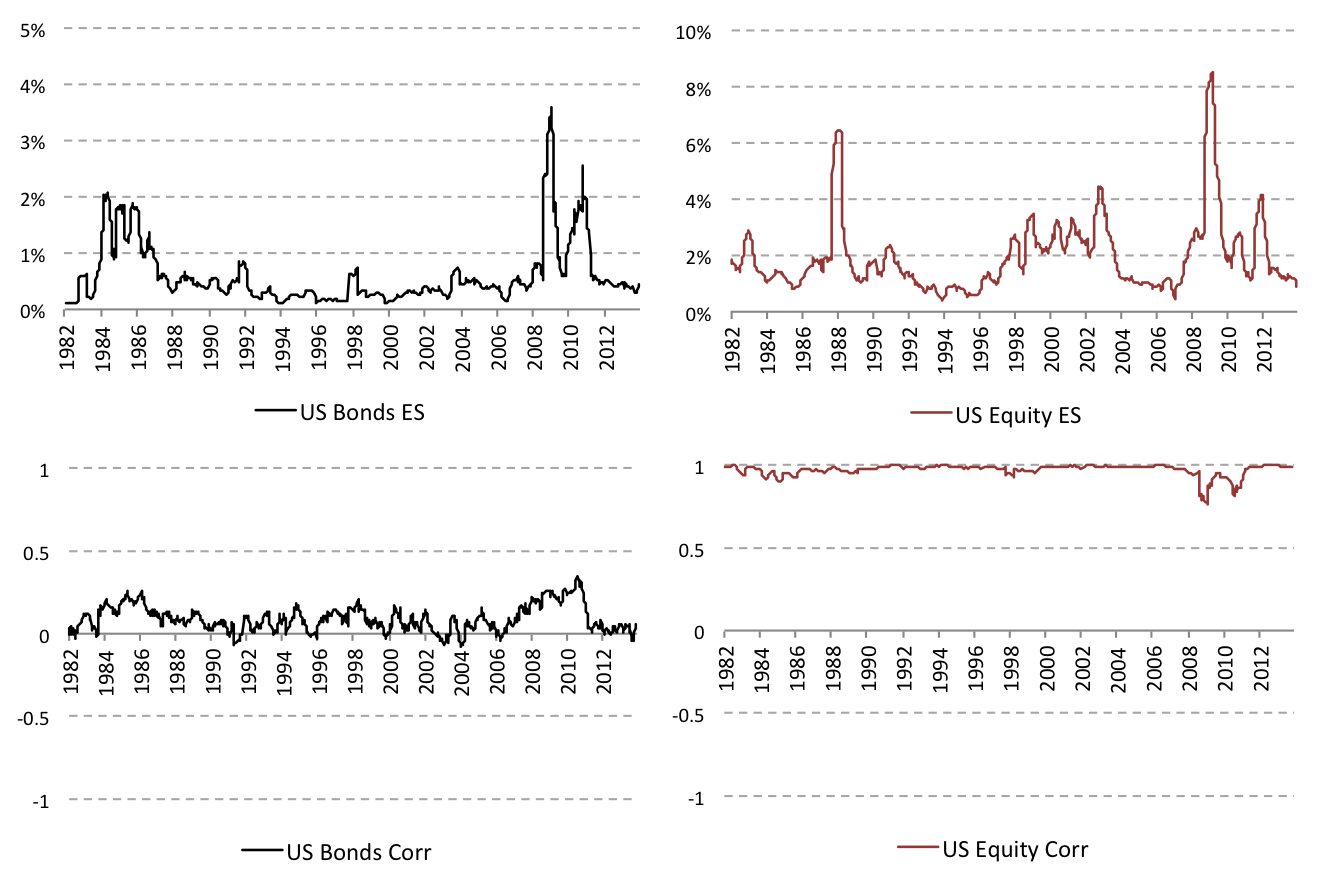}
\caption{}
\label{es_decomp_1}
\end{subfigure}
\begin{subfigure}{\textwidth}
\centering
\includegraphics[width=0.7\linewidth]{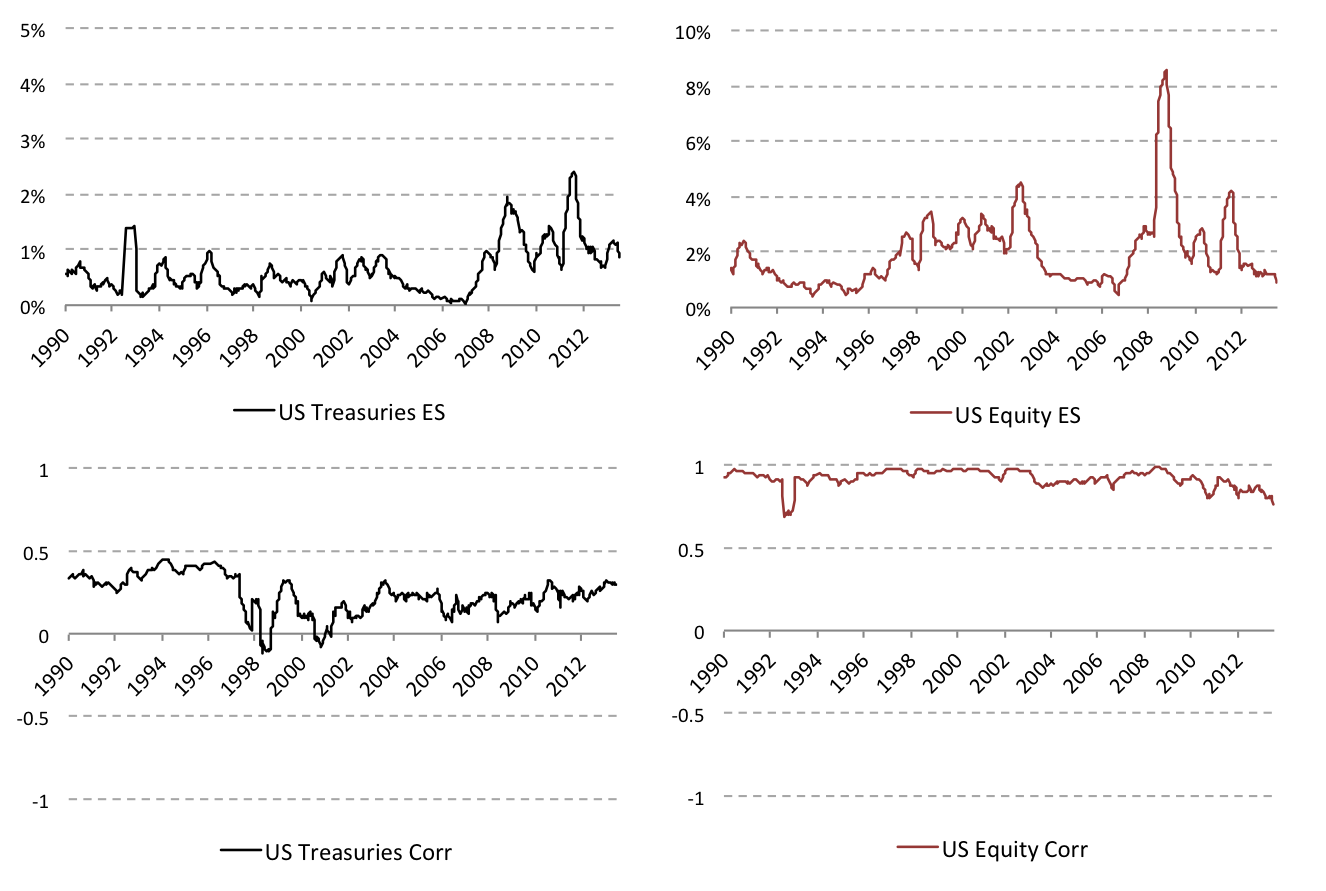}
\caption{}
\label{es_decomp_2}
\end{subfigure}
\caption{\onehalfspacing Decomposition of contributions  $ \rc^\ES_i(P) = w_i \ES(X_i) \corr^\ES_i$ to 90\% Expected Shortfall (ES)  for the 60/40 allocation to (A) US Equity and US Government Bonds, and (B) US Equity and US Treasury Bonds over the time period 1982--2013.}
\label{es_decomp}
\end{figure}

\newpage

\bibliographystyle{apalike}
\bibliography{drawdown_references}

\end{document}